\documentclass{itor}

\usepackage{natbib}%
\usepackage[figuresright]{rotating}

\usepackage{url}
\usepackage{algorithmic}
\usepackage{algorithm}
\usepackage{complexity}

\usepackage{amsmath,amsthm}

\newtheorem{theorem}{Theorem}[section]
\newtheorem{proposition}{Proposition}[section]
\newtheorem{lemma}{Lemma}[section]
\newtheorem{corollary}{Corollary}[section]

\begin{document}

\title{Simpler and efficient characterizations of tree $t$-spanners for graphs with few $P_4$'s and $(k,\ell)$-graphs}


\author[Running Author]{Fernanda Couto\affmark{a},  Lu\'is Cunha\affmark{b,$\ast\ast$} and Diego Ferraz\affmark{c}}

\affil{\affmark{a}Departamento de Ci\^encia da Computa\c{c}\~ao, Universidade Federal Rural do Rio de Janeiro, Nova Igua\c{c}u, Rio de Janeiro, Brasil}
\affil{\affmark{b}Departamento de Ci\^encia da Computa\c{c}\~ao, Universidade Federal Fluminense, Niter\'oi, Rio de Janeiro, Brasil}
\affil{\affmark{c}COPPE, Universidade Federal do Rio de Janeiro, Rio de Janeiro, Brasil}
\email{fernandavdc@ufrrj.br [F. Couto]; lfignacio@ic.uff.br [L. Cunha];\\ ferrazda@cos.ufrj.br [D. Ferraz]}



\thanks{\affmark{$\ast\ast$}Author to whom all correspondence should be addressed (e-mail: lfignacio@ic.uff.br).}

\historydate{Received DD MMMM YYYY; received in revised form DD MMMM YYYY; accepted DD MMMM YYYY}

\begin{abstract}
A tree $t$-spanner of a graph $G$ is a spanning tree $T$ in which the distance between any two adjacent vertices of $G$ is at most~$t$. The smallest $t$ for which $G$ has a tree $t$-spanner is called tree stretch index. The $t$-admissibility problem aims to decide whether the tree stretch index is at most $t$.  Regarding its optimization version, the smallest $t$ for which $G$ is $t$-admissible is the stretch index of~$G$, denoted by $\sigma_T(G)$. Given a graph with $n$ vertices and $m$ edges, the recognition of $2$-admissible graphs can be done $O(n+m)$ time, whereas $t$-admissibility is \NP-complete for $\sigma_T(G) \leq t$, $t \geq 4$ and deciding if $t = 3$ is an open problem, for more than 20 years. Since the structural knowledge of classes can be determinant to classify $3$-admissibility's complexity, in this paper we present simpler and faster algorithms to check $2$ and $3$-admissibility for families of graphs with few $P_4$'s and $(k,\ell)$-graphs. Regarding $(0,\ell)$-graphs, we present lower and upper bounds for the stretch index of these graphs and characterize graphs whose stretch indexes are equal to the proposed upper bound. Moreover, we prove that $t$-admissibility is \NP-complete even for line graphs of subdivided graphs. 
\end{abstract}


\maketitle

\section{Introduction}\label{sec:intro}

A \emph{tree $t$-spanner} of a graph $G$ is a spanning tree $T$ of $G$ in which the distance between every pair of vertices is at most $t$ times their distance in~$G$ or, equivalently, the distance between two adjacent vertices of $G$ is at most~$t$ in $T$~(\emph{cf.}~\cite{cai1995tree,ctw20,TCS20}). 
If a graph has a tree $t$-spanner, then it is called a \emph{tree $t$-spanner admissible graph} (or simply \emph{$t$-admissible}). 
The parameter~$t$ of a tree $t$-spanner is called the \emph{tree stretch factor}, denoted by $\sigma(T)$, and the smallest $t$ for which $G$ is~$t$-admissible is the \emph{tree stretch index} of $G$, denoted by~$\sigma_T(G)$. 
The $t$-admissibility problem aims to decide whether $\sigma_T(G)\leq t$. 
The problem of determining the \emph{tree stretch index} of $G$ is also called the \textsc{minimum stretch spanning tree problem} (MSST). 
From now on, when we refer to MSST, we are dealing with the decision version of this problem. 

\subsection*{On the $2$-admissibility problem}

Cai and Corneil~\cite{cai1995tree} proved that $t$-admissibility is \NP-complete, for $t \geq 4$, whereas $2$-admissible graphs can be recognized in $O(|V|+|E|)$ time. 
However, the characterization of $3$-admissible graphs is still an open problem. 
The characterization of $2$-admissible graphs (Theorem~\ref{thm:caicorneil}) deals with triconnected components of a connected graph, defined as any maximal subgraph that does not contain two vertices whose removal disconnects the graph. 
A \emph{nonseparable} graph is a graph without a \emph{cut vertex}, i.e., a vertex whose removal disconnects the graph. A \emph{star} with $n+1$ vertices is the complete bipartite graph $K_{1,n}$. A \emph{$v$-centered star} is a star centered on~$v$.

\begin{theorem}\cite{cai1995tree}\label{thm:caicorneil}
A nonseparable graph $G$ is $2$-admissible if and only if $G$ contains a spanning tree $T$ such that for each triconnected component $H$ of $G$, $T\cap H$ is a spanning star of~$H$.
\end{theorem}

So far, there are not studies regarding structural characterization of $2$-admissible graphs. Theorem~\ref{thm:caicorneil} can be used in order to develop an $O(n+m)$ time algorithm (see~\cite{cai1995tree}), which consists, in a nutshell, of: i) finding universal vertices in each triconnected component of a graph $G$; or ii) verifying that in each biconnected component of $G$ with adjacent vertices $u$ and $v$, forming a vertex cut of length two, the corresponding edge $uv$ belongs to a tree $2$-spanner of $G$. 
The resulted subgraph $H$ considering i) and ii) is a tree $2$-spanner of $G$ if and only if $H$ is a spanning tree of $G$.

The above strategy is a constructive algorithm and do not provide directly structural characterizations for graphs to be $2$-admissible without considering their spanning tree. 
Since the recognition of $3$-admissible graphs is an open problem, there are graph classes whose stretch indexes are bounded by specific values, such as cographs or split graphs (\emph{cf.}~\cite{TCS20}). 
Moreover, there are not many studies developing algorithms in order to determine the exact value of the stretch indexes for graph classes which are $3$-admissible in a better way than the recognition of general $2$-admissible graphs.



\subsection*{Contributions}

We develop faster and simpler ways to determine stretch indexes for: 
graphs with few $P_4$'s, i.e. graphs with a bounded number of induced $P_4$'s; 
and 
subclasses of $(k,\ell)$-graphs, i.e. graphs whose vertex sets can be partitioned into $k$ independent sets and $\ell$ cliques. 

Both classes have already been considered by~\cite{LAGOS19,TCS20}, by an extensive study on \P \ \emph{versus} \NP-complete dichotomy regarding $t$-admissibility. 
Section~\ref{sec:P4} is devoted to study of admissibility for graphs with few $P_4$'s, which have stretch index equal to $2$ or $3$. 
The proposed strategies are simpler and faster than the application of Cai and Corneil's recognition algorithm for cographs, $P_4$-sparse graphs and $P_4$-tidy graphs. 
Section~\ref{sec:kl} presents structural characterizations for $3$-admissible $(k,\ell)$-graphs, such as $(1,1)$- and $(0,2)$-graphs. 
Moreover, regarding $(0,\ell)$-graphs, lower and upper bounds for the stretch index are presented, as well as a characterization of graphs that have stretch indexes equal to the proposed upper bound. 
We also settle that $t$-admissibility is \NP-complete for line graphs of subdivided graphs, 
which turns $t$-admissibility closely related to the edge $t$-admissibility problem, as proposed by~\cite{ctw20}. 
Section~\ref{sec:conc} concludes the paper with final remarks and open questions.

\paragraph{Preliminaries}\label{sub:pre}

Given a graph $G=(V,E)$, $d_G(x,y)$ denotes the distance between $x$ and $y$ in~$G$ and $d_G(v)$, the degree of~$v$ in $G$. A \emph{pendant vertex} is a vertex of degree~$1$.
A \emph{non-edge} of a spanning tree $T$ is an edge of $G\setminus T$. 
A $p$-path (or a path $P_{p+1}$) is a path formed by $p$ edges ($p+1$ vertices). 
Given a graph $G$, its \emph{line graph} $L(G)$ is obtained as follows: $V(L(G))=E(G)$; $E(L(G)) = \{\{uv,uw\} | uv, uw \in E(G)\}$. I.e., each edge of $G$ is a vertex of $L(G)$ and if two edges share an endpoint, then their corresponding vertices are adjacent in $L(G)$. 
The \emph{distance between two edges} $e_1$ and $e_2$ of $G$, for $e_1,e_2 \in E(G)$ is the distance between their corresponding vertices in $L(G)$. 
The subdivided graph of $G$ is the one so that each edge $vw \in E(G)$ becomes a path $v,x,w$, where $x$ is a new vertex. 

\section{Graphs with few $P_4$'s}\label{sec:P4}
Graphs with few $P_4$'s can be constructed by a finite number of operations, as union and join. 
Given graphs $G_i = (V_i,E_i), \ i= 1, \ldots, p$, we formally define the union and the join operations, resp., as follows: $G_1 \ \textcircled{\footnotesize 0} \ \cdots \textcircled{\footnotesize 0} \ G_p = (V_1\cup \cdots \cup V_p, E_1 \cup \cdots \cup E_p)$; $G_1 \ \textcircled{\footnotesize 1} \ \cdots \ \textcircled{\footnotesize 1} \ G_p = (V_1 \cup \cdots \cup V_p, E_1 \cup \cdots \cup E_p \cup \{xy \ | \ x \in V_i, y \in V_j, \ i \neq j, \ 1 \leq i,j \leq p\})$.

A \emph{cograph} is a $P_4$-free graph. A trivial graph is a cograph, and any other can be obtained by disjoint union or join operations of cographs. 
Lemma~\ref{lm:join} states that any graph obtained by the join of two graphs, $G = G_1 \textcircled{\footnotesize 1} G_2$, is a $3$-admissible graph.

\begin{lemma}\label{lm:join}[\cite{LAGOS19}]
 Given graphs $G_1$ and $G_2$, and $G = G_1 \textcircled{\footnotesize 1} G_2$, we have that $G$ is $3$-admissible.
\end{lemma}

As a direct consequence of Lemma~\ref{lm:join}, we have that cographs are $3$-admissible. Moreover, Lemma~\ref{lm:join} implies that, if $G$ is a cograph, then $\sigma_T(G) \leq 3$. Interestingly, having a universal vertex is also a necessary condition for a cograph to be $2$-admissible, according to~\cite{LAGOS19}. 
Hence, $2$-admissibility recognition can be done in $O(n)$ time for cographs, for $n$ being the number of vertices of a given cograph $G$, 
by checking whether $G$ has a universal vertex. 

A graph is \emph{$P_4$-sparse} if for each set of $5$ vertices, there is at most one induced $P_4$. 
A graph $G$ is \emph{$P_4$-tidy} if, for each induced $P_4$ of $G$, say $P$, there is at most one vertex $v \in V(G) \setminus V(P)$ such that $V(P)\cup \{v\}$ induces at most two $P_4$'s in $G$.

\paragraph{Stretch index \emph{vs} spider operation}

A graph $G$ is a \emph{spider} if its vertex set can be partitioned into $\mathcal{S},\mathcal{K}$ and $\mathcal{R}$ such that (i) $\mathcal{K}$ is a clique, $\mathcal{S}$ is an independent set and $|\mathcal{S}|=|\mathcal{K}| \geq 2$; (ii) each vertex of $\mathcal{R}$ is adjacent to all vertices of $\mathcal{K}$ (a join operation) and is non-adjacent to any vertex of $\mathcal{S}$; (iii) There is a bijection $f: \mathcal{S} \mapsto \mathcal{K}$ such that, for all $x \in \mathcal{S}$, either $N(x) = \{f(x)\}$, called a \emph{thin spider}, or $N(x) = \mathcal{K} - \{f(x)\}$, called a \emph{thick spider}. 


\cite{jamison1992tree} constructively characterized $P_4$-sparse graphs. A graph $G$ is \emph{$P_4$-sparse} if and only if for each one of its induced subgraphs $H$, exactly one of the following conditions is satisfied: (i) $H$ is disconnected; (ii) $\overline{H}$ is disconnected; (iii) $H$  is isomorphic to a spider. Note that items (i) and (ii) suggest the union and the join operations applied in a cograph construction. 
In order to construct a $P_4$-sparse graph, an operation concerning item (iii) is defined in the following. Let $G_1 = (V_1, \emptyset)$ and $G_2=(V_2,E_2)$ be two disjoint graphs, where $V_2 = \{v\} \cup \mathcal{K} \cup \mathcal{R}$ and such that: (a) $|\mathcal{K}| = |V_1| + 1 \geq 2$; (b) $\mathcal{K}$ is a clique; (c) $x \in \mathcal{R}$ is adjacent to each vertex $x' \in \mathcal{K}$ and $x$ is not adjacent to $v$; (d) there exists a vertex $v' \in \mathcal{K}$ such that $N_{G_2}(v) = \{v'\}$ or $N_{G_2}(v) = \mathcal{K} \setminus \{v'\}$. Choose a bijective function $f: V_1 \mapsto \mathcal{K}\setminus \{v'\}$ and define the operation \textcircled{\footnotesize 2} as follows: $G_1 \ \textcircled{\footnotesize 2} \ G_2 = (V_1 \cup V_2, E_2 \cup E')$, where $E' = \{xf(x) \ | \  x \in V_1\}, \ \mbox{if} \ N_{G_2}(v) = \{v'\}$, or $E' = \{xz \ | \ x \in V_1, z \in \mathcal{K} \setminus \{v'\}\}, \ \mbox{if} \ N_{G_2}(v) = \mathcal{K} \setminus \{v'\}$. 

A graph is a spider if and only if it can be obtained by the two proper induced subgraphs generated by~\textcircled{\footnotesize 2}. 
Moreover, a spider is $P_4$-sparse if and only if the subgraph induced by $\mathcal{R}$ is $P_4$-sparse~\cite{jamison1992tree}.
In this way, a graph $G$ is $P_4$-sparse if and only if $G$ can be obtained from trivial graphs, by applying, in any order, operations \textcircled{\footnotesize 0}, \textcircled{\footnotesize 1} and \textcircled{\footnotesize 2} a finite number of times.

As a consequence, each $P_4$-sparse graph has an associated tree, called \emph{PS-tree}. Essentially, in a PS-tree, leaves are the vertices of the graph, each internal node is labeled by $0,1$ or $2$ (accordingly to the operation applied to the associated subtree). 
See~\cite{jamison1992tree} for construction details.


\begin{lemma}\label{lm:spider-index}[\cite{LAGOS19}]
 Let $G$ be a spider graph. If $G$ is a thin spider, then $\sigma_T(G) = 2$. Otherwise, $\sigma_T(G) = 3$.
\end{lemma}

\begin{lemma}\label{lm:p4sparse-upper}
Let $G$ be a $P_4$-sparse graph, then $\sigma_T(G)\leq 3$.
\end{lemma}
\begin{proof}
Let $T$ be the $PS$-tree of $G$ rooted in $r$. Since $G$ must be connected, then $r$ has label either $1$ or $2$. In each case, we analyze the possible root's subtrees. If $r$'s label is equal to $1$, then there is a \textcircled{\footnotesize 1} operation, and by Lemma~\ref{lm:join}, we have that $G$ is $3$-admissible. If $r$'s label is equal to $2$, then $G$ is a spider graph, and $3$-admissible by Lemma~\ref{lm:spider-index}.
\end{proof}


\begin{lemma}\label{lm:oct}[\cite{LAGOS19}]
 If $G$ is a connected $P_4$-sparse graph which is not a thin spider and without a universal vertex, then $\sigma_T(G) = 3$. 
\end{lemma}

\begin{theorem}\label{thm:p4sparse-index}[\cite{LAGOS19}]
A $P_4$-sparse graph $G$ is $2$-admissible if and only if either $G$ has a universal vertex; or $G$ is a thin spider.
\end{theorem}
\begin{proof}
Clearly, if $G$ has a universal vertex or if $G$ is a thin spider, then $\sigma_T(G)=2$. For the converse, suppose $G$ is not a thin spider and does not have a universal vertex. So, its $PS$-tree's root has label $2$ (in this case $G$ is a thick spider) or $1$. Hence, by Lemmas~\ref{lm:spider-index} and ~\ref{lm:oct}, resp., $\sigma_T(G)=3$.~\end{proof}


\cite{giakoumakis1997p4} described a recognition algorithm to spider graphs and the corresponding triple ($\mathcal{S},\mathcal{K},\mathcal{R}$) to spider partitions in linear time. 
Hence, given a $P_4$-sparse graph, based on Theorem~\ref{thm:p4sparse-index}, we have the following theorem.

\begin{theorem}\label{thm:p4sparse-complex}
$2$-admissibility can be decided in $O(n)$ time for $P_4$-sparse graphs with $n$ vertices.
\end{theorem}
\begin{proof}
By the degree sequence of a $P_4$-sparse graph $G$, we can check if there is a universal vertex or if there are pendant vertices associated to the independent set $\mathcal{S}$ of a thin spider.
\end{proof}

A natural generalization of $P_4$-sparse graphs are the $P_4$-tidy graphs. A graph $H$ is an \emph{almost-spider} graph if $H$ can be constructed from a spider graph $G = (\mathcal{S}, \mathcal{K}, \mathcal{R})$ by adding a vertex $v'$ which is either a false twin of $v$ or a true twin of $v$, such that $v \in \mathcal{S} \cup \mathcal{K}$~\cite{jamison1995p}. Hence, we call $H$ a \emph{$\mathcal{P}$-false-almost-spider} and \emph{$\mathcal{P}$-true-almost-spider}, resp., where $\mathcal{P}$ is the set to which $v$ belongs, i.e, $\mathcal{P} \in \{\mathcal{S}, \mathcal{K}\}$. In the same way, if $G$ is a thin (or thick) spider, then $H$ is a \emph{true} or \emph{false-almost-thin (or thick)-spider}. A \emph{$P_4$-tidy} graph $G$ can be constructed by the following way: i) $G_1$ \textcircled{\footnotesize 0} $G_2$, for $G_1$ and $G_2$ being $P_4$-tidy graphs; ii) $G_1$ \textcircled{\footnotesize 1} $G_2$, for $G_1$ and $G_2$ being $P_4$-tidy graphs; iii) $G$ is a spider; iv) $G$ is an almost spider; v) $G$ is $P_5$, $C_5$, $\overline{P_5}$, or $K_1$. Since $PS$-trees represent $P_4$-sparse graphs, we can develop in a similar way a tree representation of a $P_4$-tidy graph~\cite{giakoumakis1997p4}.

\begin{lemma}
 Let $G$ be an almost-spider graph, then $\sigma_T(G)\leq 3$.
\end{lemma}
\begin{proof}
Fig.~\ref{fig:p4s-pstree} depicts all almost-spider graphs and their solutions.
\begin{figure}[!h]
    \centering
    \hspace{-0.25cm}\raisebox{2cm}{a)}\includegraphics[scale=0.14]{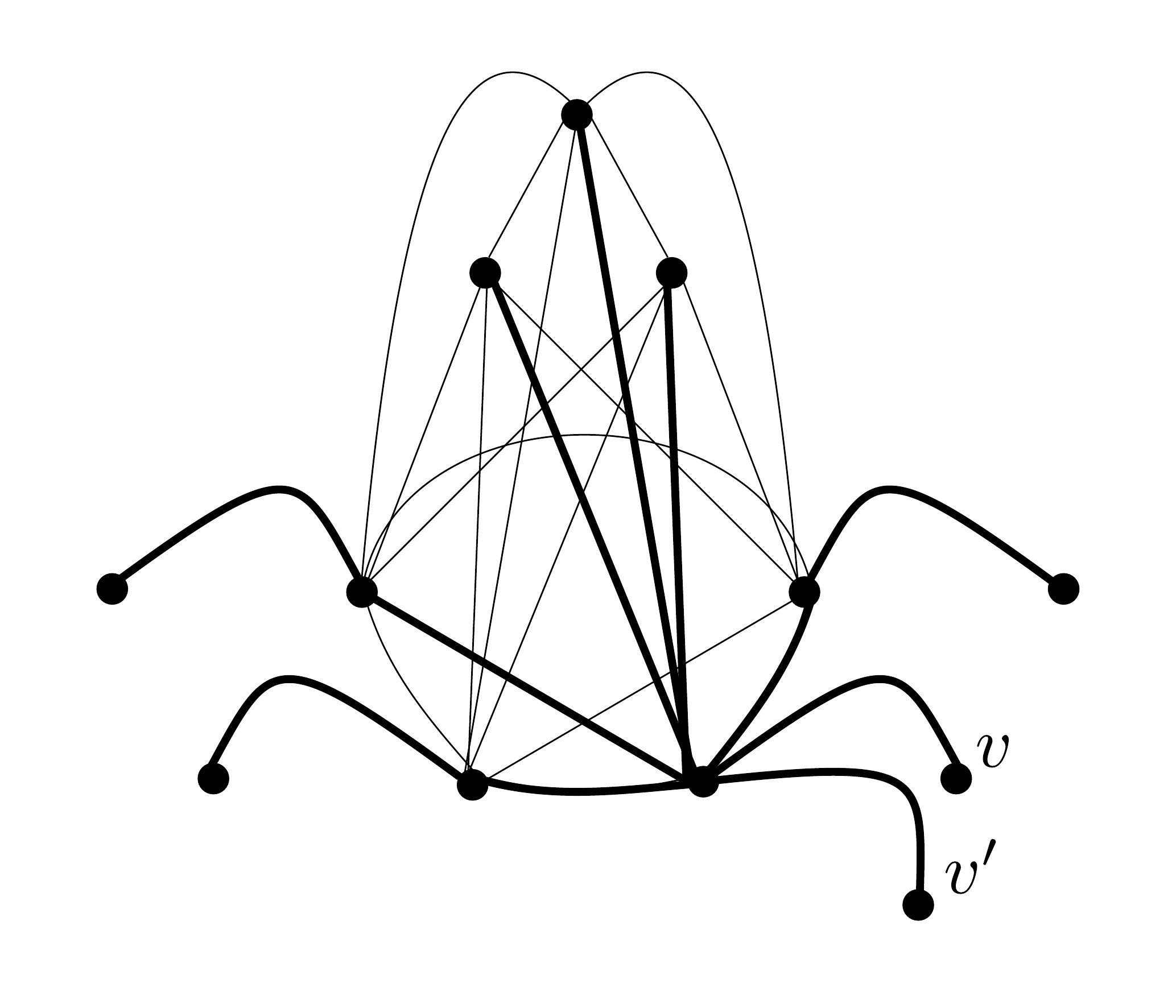}\hspace{-0.35cm}
    \raisebox{2cm}{b)}\includegraphics[scale=0.14]{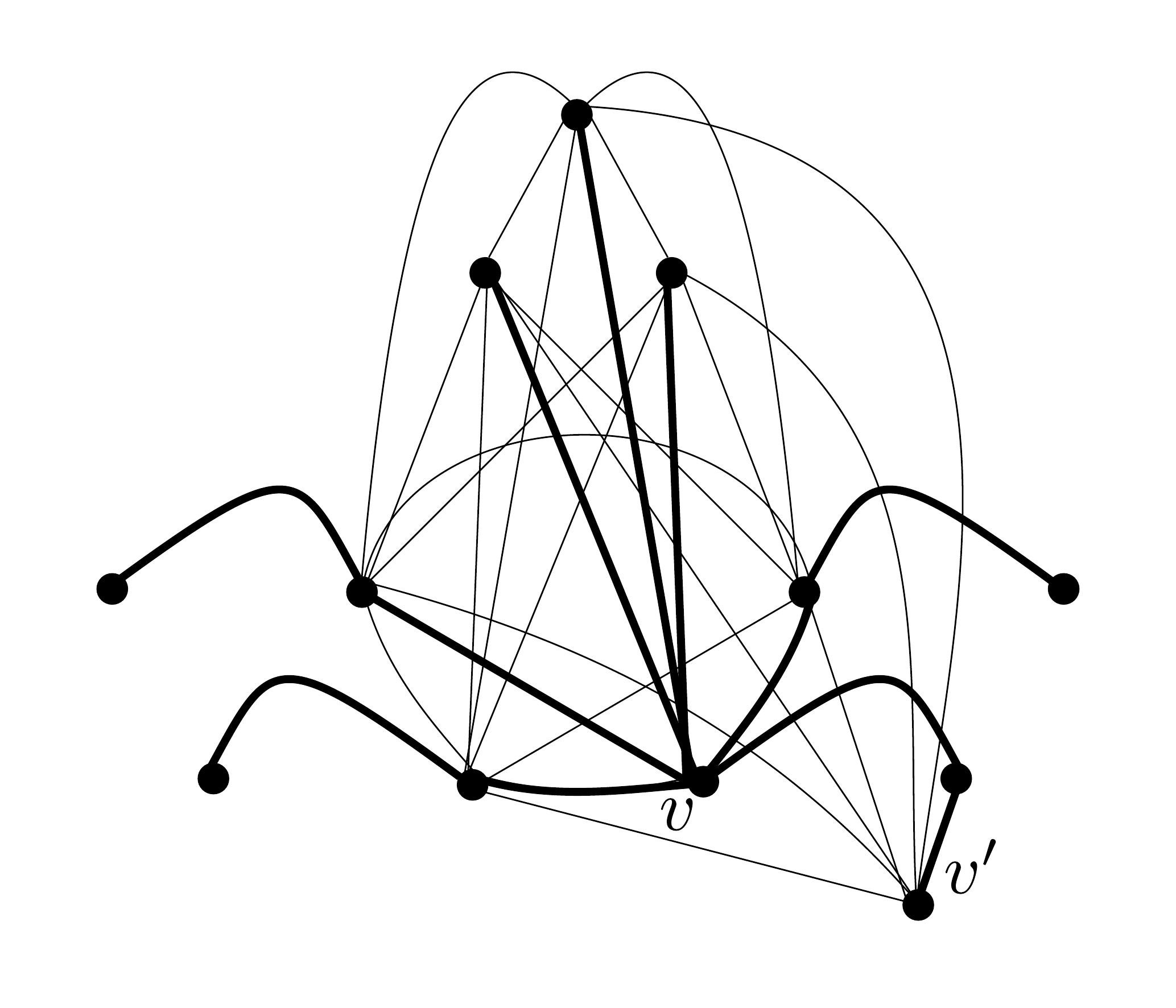}\hspace{-0.35cm}
    \raisebox{2cm}{c)}\includegraphics[scale=0.14]{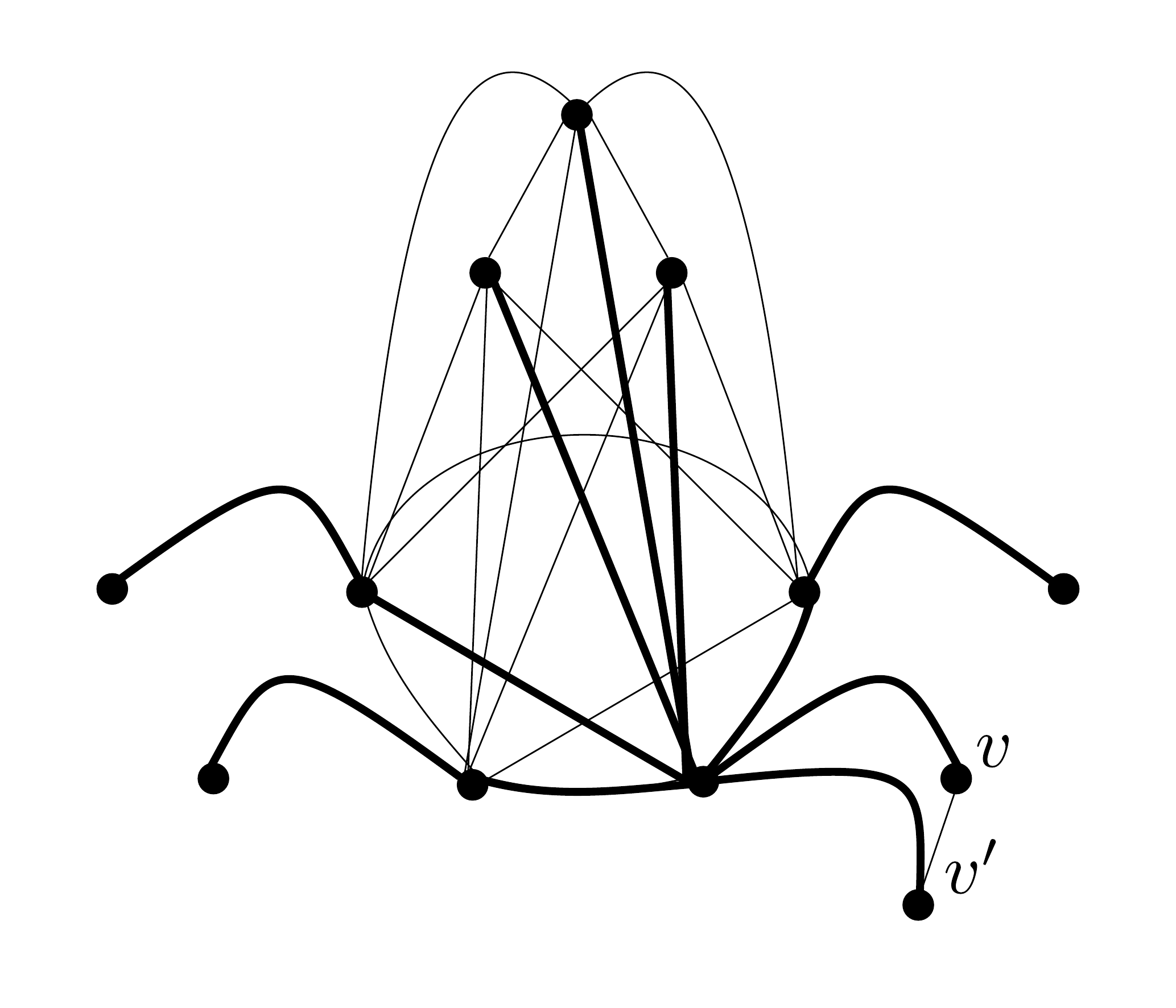}\hspace{-0.35cm}
    \raisebox{2cm}{d)}\includegraphics[scale=0.14]{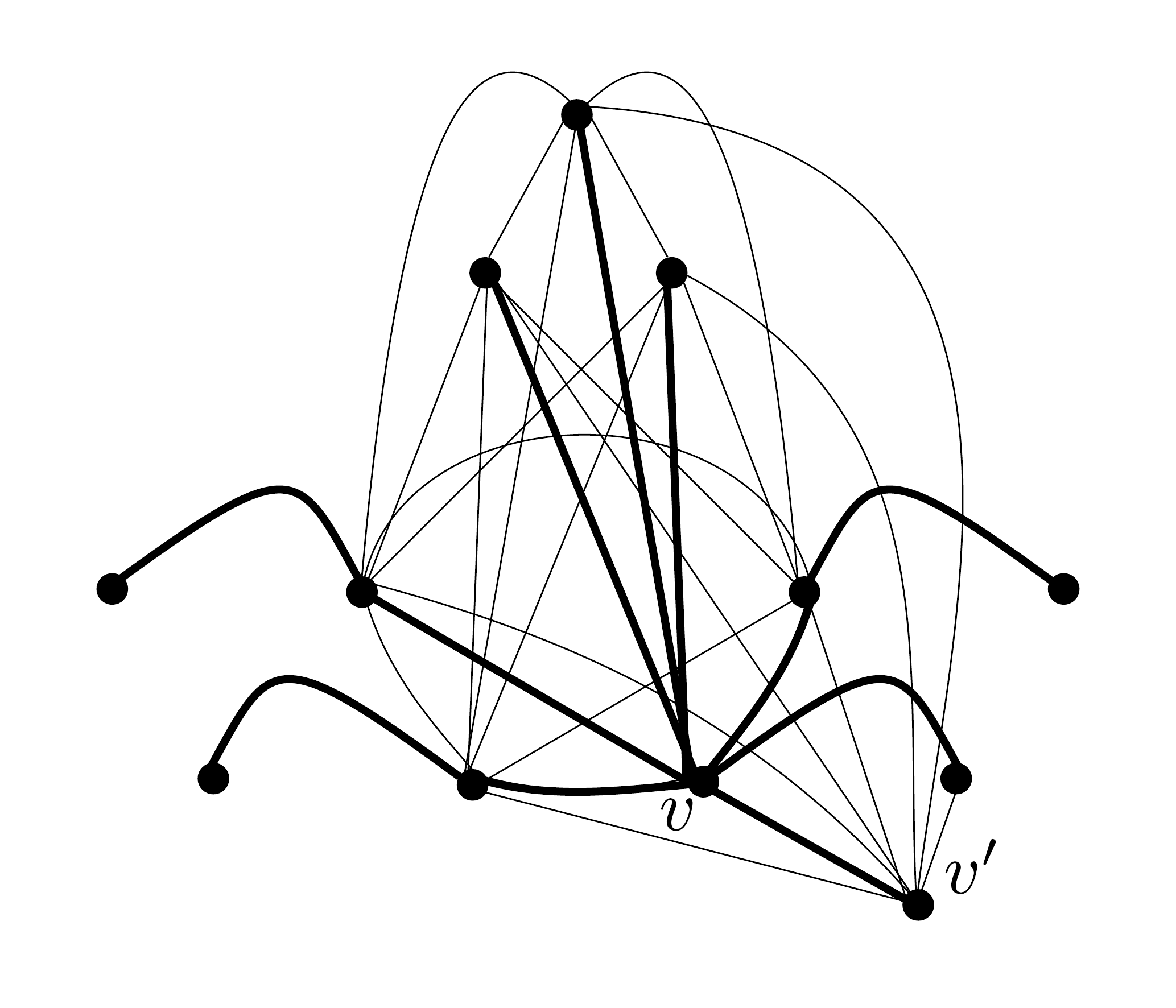}
    
    \hspace{-0.25cm}\raisebox{2cm}{e)}\includegraphics[scale=0.14]{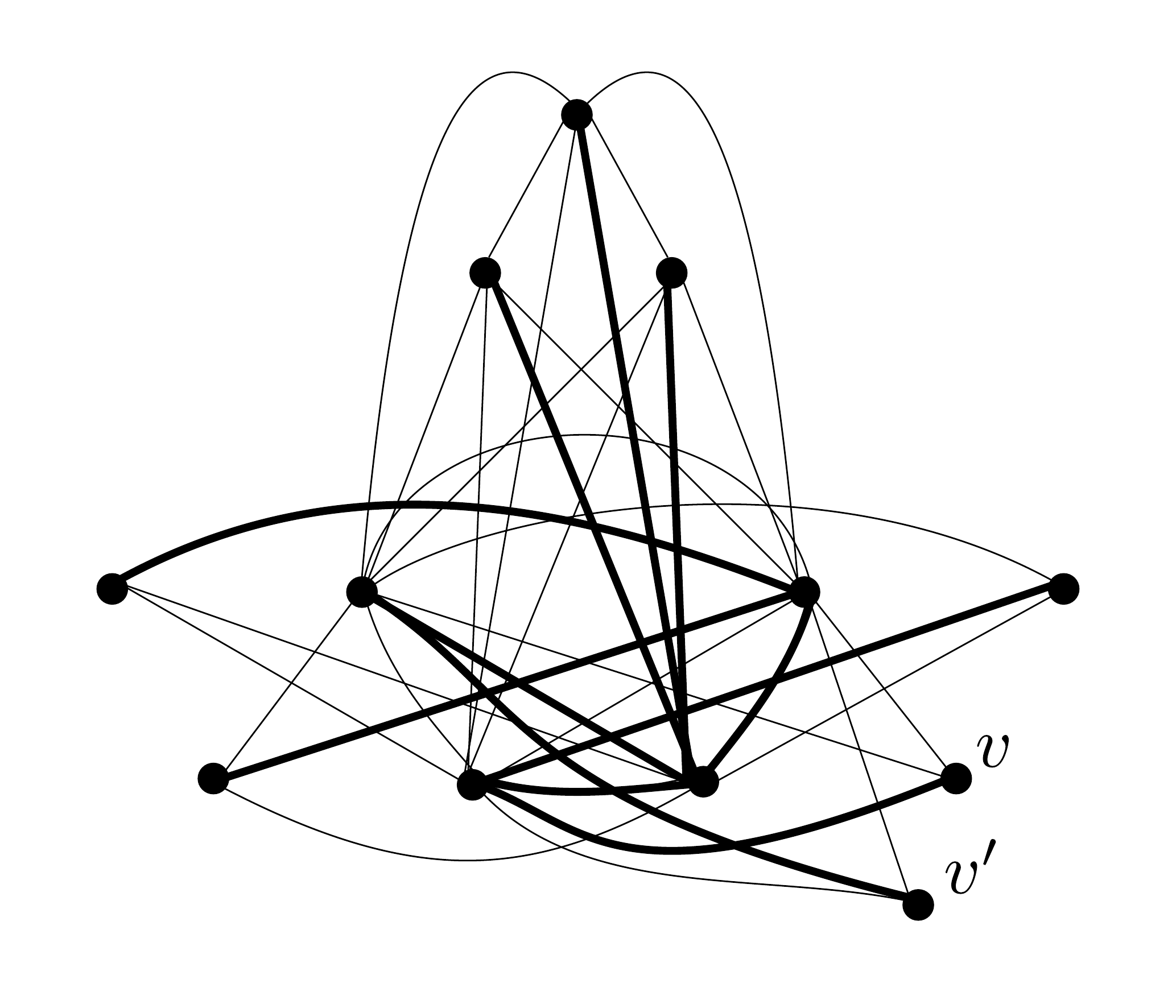}\hspace{-0.35cm}
    \raisebox{2cm}{f)}\includegraphics[scale=0.14]{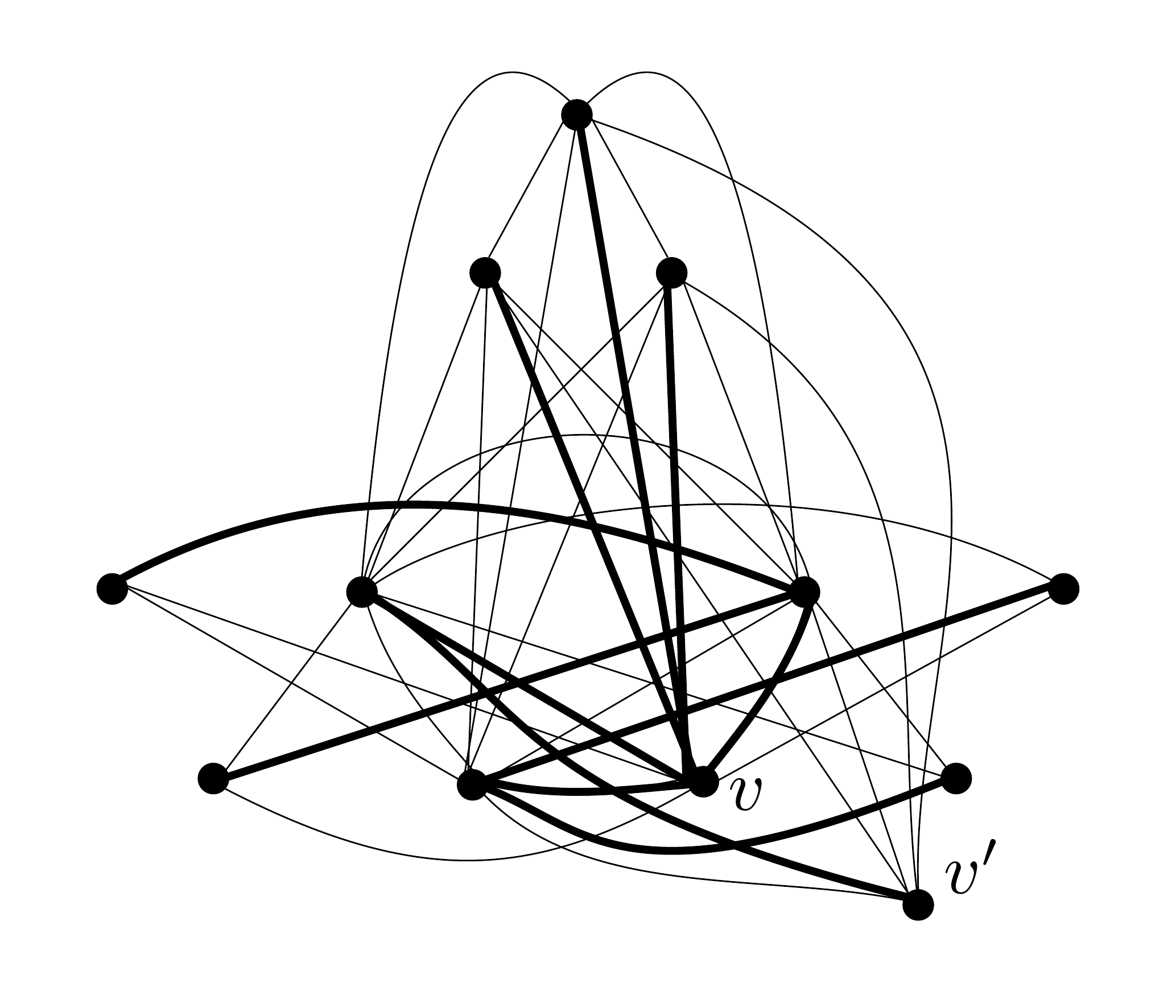}\hspace{-0.35cm}
    \raisebox{2cm}{g)}\includegraphics[scale=0.14]{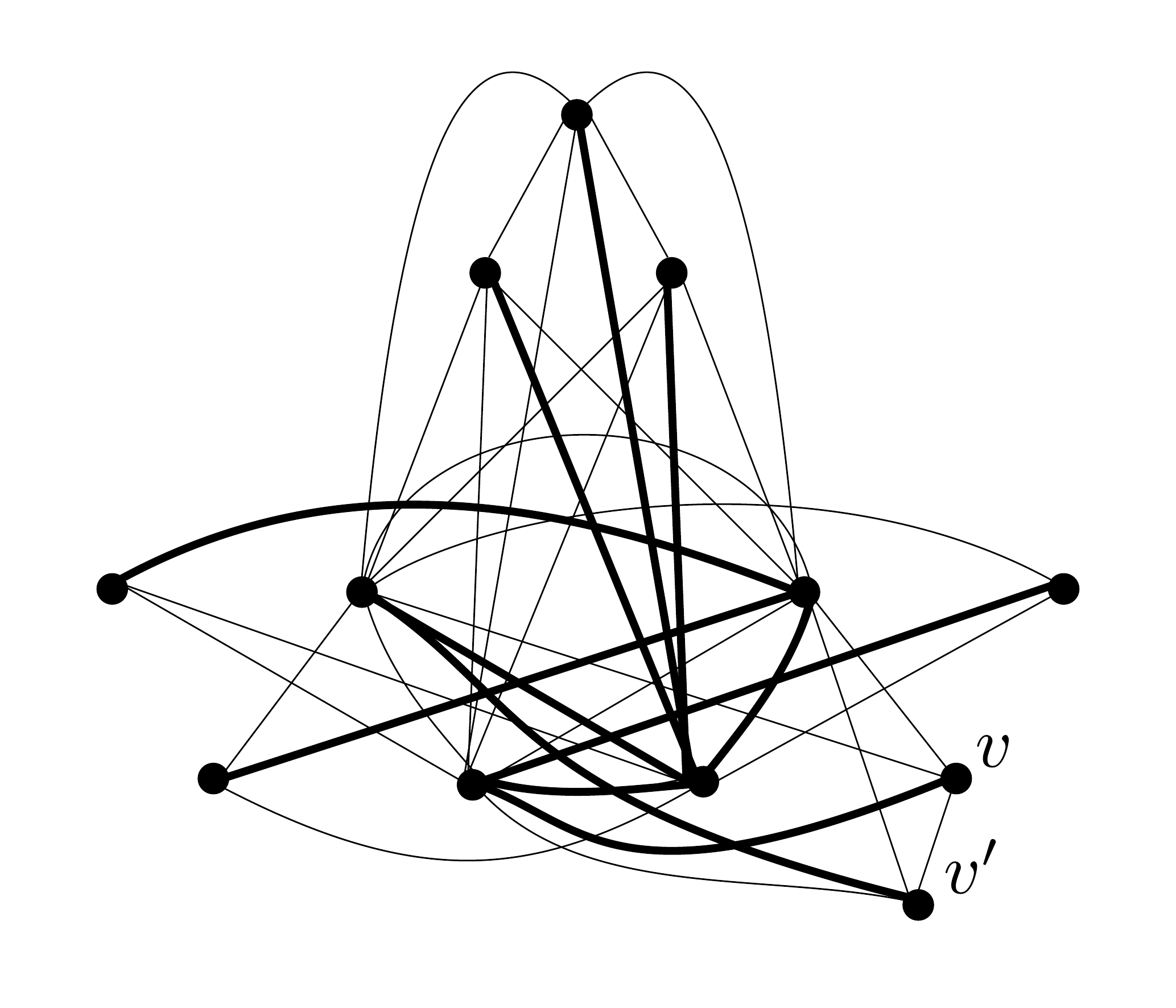}\hspace{-0.35cm}
    \raisebox{2cm}{h)}\includegraphics[scale=0.14]{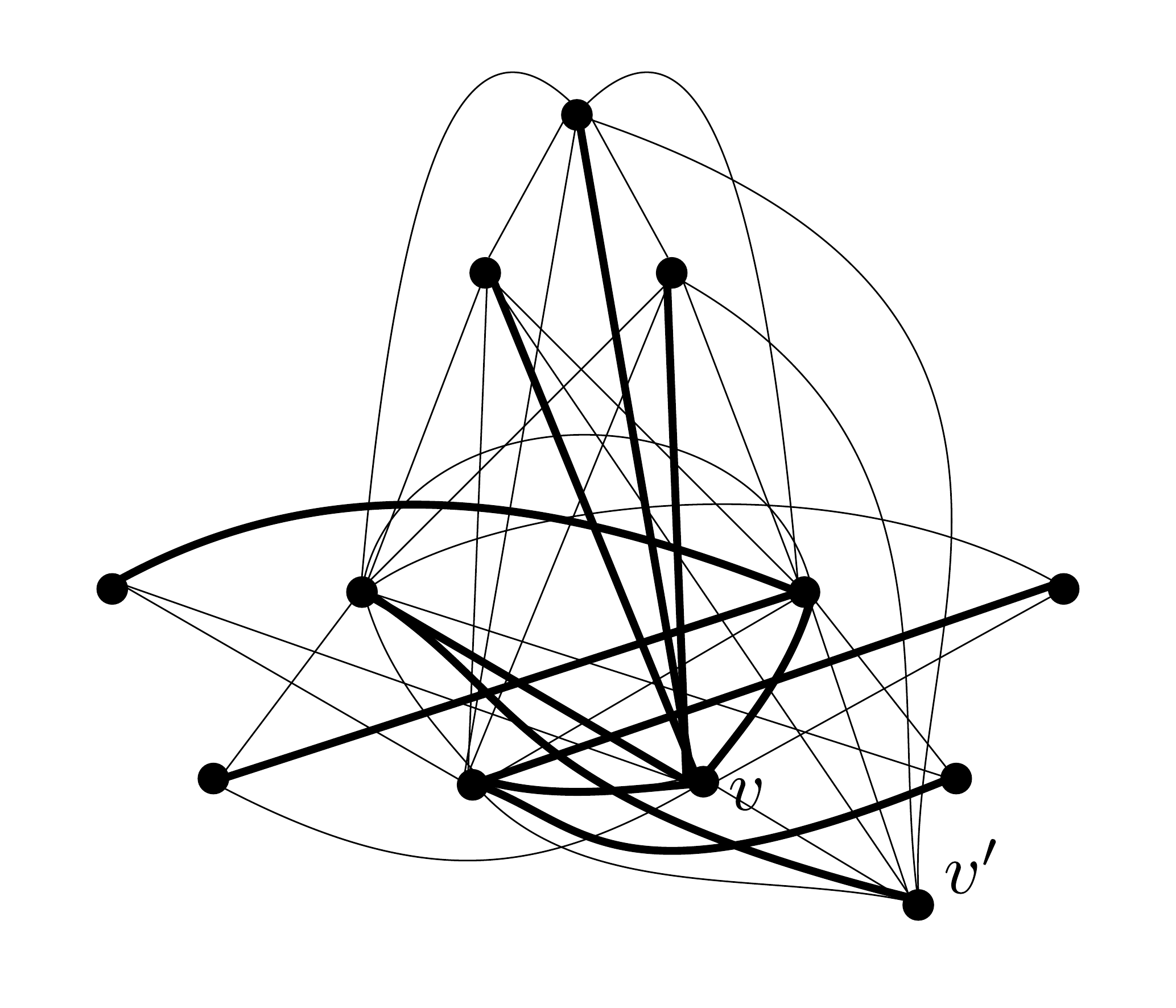}\vspace{-.4cm}
    \caption{Almost spider graph $H$ obtained from spider graph $G$, $v$ and $v'$ are twin vertices. Bold edges are 
    tree spanners of $H$. 
    a) $H$ is a $\mathcal{S}$-false-almost-thin-spider. $\sigma_T(G) = 2$ and $\sigma_T(H) = 2$; b) $\mathcal{K}$-false-almost-thin-spider. $\sigma_T(G) = 2$ and $\sigma_T(H) = 3$; c) $\mathcal{S}$-true-almost-thin-spider. $\sigma_T(G) = 2$ and $\sigma_T(H) = 2$; d) $\mathcal{K}$-true-almost-thin-spider. $\sigma_T(G) = 2$ and $\sigma_T(H) = 2$; e) $\mathcal{S}$-false-almost-thick-spider. $\sigma_T(G) = 3$ and $\sigma_T(H) = 3$; f) $\mathcal{K}$-false-almost-thick-spider. $\sigma_T(G) = 3$ and $\sigma_T(H) = 3$; g) $\mathcal{S}$-true-almost-thick-spider. $\sigma_T(G) = 3$ and $\sigma_T(H) = 3$; h) $\mathcal{S}$-true-almost-thick-spider. $\sigma_T(G) = 3$ and $\sigma_T(H) = 3$.}
    \label{fig:p4s-pstree}
\end{figure}
\end{proof}

By cases i) - v) to construct a $P_4$-tidy graph and results above, if $G$ is a $P_4$-tidy graph, then $G$ is $3$-admissible, excepted if $G$ is a $C_5$.

\begin{lemma}\label{lm:p4tidy-upper}
Let $G$ be an almost-spider graph. $G$ is $2$-admissible if and only if $G$ is a $\mathcal{S}$-almost-thin-spider (false or true), or is a $\mathcal{K}$-true-almost-thin-spider.
\end{lemma}
\begin{proof}
Clearly, if $G$ is an almost spider, as described above, it is $2$-admissible. Conversely, suppose $G$ is a $2$-admissible almost-spider, and, by contradiction, that $G$ is either an almost-thick-spider or a $\mathcal{K}$-false-almost-thin-spider. It is easy to see that if $G$ is an almost-thick-spider the stretch index is at least $3$, since the addition of a true or false twin to a thick spider cannot decrease its parameter. So, suppose $G$ is a $\mathcal{K}$-false-almost-thin-spider by the addition of a false twin $v$ of $u \in \mathcal{K}$. Note that $H= \mathcal{K} \cup \mathcal{R} \cup \{v\}$ is a triconnected component. Since $G$ is $2$-admissible, there is a spanning tree of $G$ such that $T \cap H$ is a star. Since $v$ nor $u$ are universal in $H$, they are leaves of the star. Thus the vertex in $\mathcal{S}$ that is adjacent to both is apart from one of them in $T$ by a distance of at least $3$, a contradiction. 
\end{proof}

Similarly to Theorem~\ref{thm:p4sparse-index}, we are able to characterize the $P_4$-tidy graphs that are $2$-admissible. 

\begin{theorem}
A $P_4$-tidy graph $G$ is $2$-admissible if only only if either: $G$ has a universal vertex; or $G$ is a thin spider; or $G$ is a $\mathcal{S}$-almost-thin-spider (false or true); or $G$ is a $\mathcal{K}$-true-almost-thin-spider.
\end{theorem}
\begin{proof}
If $G$ is one of the graphs described above, then $G$ is $2$-admissible by Lemmas~\ref{lm:spider-index} and~\ref{lm:p4tidy-upper}. Conversely, suppose $G$ is a $P_4$-tidy graph $2$-admissible but it is distinct of each graph above. Hence, we have that $G$ does not have a universal vertex and it is of the following cases: i) $G$ is $\mathcal{K}$-false-almost-thin-spider; ii) $G$ is $\mathcal{S}$-almost-thick-spider (false or true); iii) $G$ is a thick spider; iv) $G$ is a join of two $P_4$-tidy graphs. If $G$ is one of cases i)-iii), then, by Lemmas~\ref{lm:oct} and~\ref{lm:p4tidy-upper}, we have a contradiction. If $G$ is a join of two $P_4$-tidy graphs $G_1$ and $G_2$, then: if $G_1$ or $G_2$ is connected, then $G$ is a triconnected graph, and since $G$ does not have a universal vertex, then $\sigma_T(G)\geq3$, a contradiction; if both $G_1$ and $G_2$ are disconnected and $G$ is a triconnected graph, then similarly to the previous case, $\sigma_T(G)\geq3$. Now, if $G_1$ and $G_2$ are disconnected and $G$ is not a triconnected graph, then $\min\{|V(G_1)|,|V(G_2)|\}=2$, and since there is not a universal vertex in $G$, then there is an induced $C_k$, $k\geq4$, in $G$, which implies $\sigma_T(G)\geq3$, a contradiction.
\end{proof}

\begin{theorem}\label{thm:p4tidy-complex}
$2$-admissibility can be decided in $O(n)$ time for $P_4$-tidy graphs with $n$ vertices.
\end{theorem}
\begin{proof}
By the degree sequence of a $P_4$-tidy graph $G$, we can check if $G$ has universal vertices in $O(n)$ time. 
Additionally, from the recognition of thin spider graphs we can check if $G$ is a thin spider, by looking for pendant vertices associated to the independent set $\mathcal{S}$ of $G$.
Moreover, we can also check if $G$ is a $\mathcal{S}$-almost-thin-spider (false or true) or if it is a $\mathcal{K}$-true-almost-thin-spider in $O(n)$ time, by checking if exactly one of the vertices in $\mathcal{K}$ has a degree greater than the number of the other vertices of $\mathcal{K}$.
\end{proof}

\section{Stretch index for $(k,\ell)$-graphs}\label{sec:kl}



\cite{TCS20} settled the \NP-completeness of MSST for $(k,\ell)$-graphs, for $k + \ell \geq 3$. 
Regarding some other values of $k$ and $\ell$, 
there is no characterization on the 2-admissibility of $(0,2)$-graphs. 
$(k,\ell)$-graphs fit on the framework of MSST interesting classes, since $(0,2)$-graphs are $3$-admissible whereas for $(2,0)$-graphs the MSST is known to be \NP-complete for $t\geq 5$~\cite{brandstadt2007tree}. 
Although for several values of $k$ and $\ell$, MSST is \NP-complete, we can determine the stretch index for some $(k,\ell)$-graphs in polynomial-time.

\subsection{$(1,1)$-graphs}

A graph $G=(X,Y)$ is a \emph{split graph}, also called a $(1,1)$-graph, if and only if it can be partitioned into a clique $X$ and a stable set~$Y$. In terms of forbidden subgraphs, they are $\{2K_2, \ C_4, C_5\}$-free~graphs.

\begin{lemma}\label{lm:upperSplit}
 If $G$ is a split graph, then $\sigma_T(G) \leq 3$.
\end{lemma}
\begin{proof}
 We obtain a spanning tree $T$ for a split graph $G = (X,Y)$ as follows. Set any vertex $x$ in $X$ to be the center of a star which includes each other vertex of $X$. Next, for each vertex $y \in Y$, choose an edge incident to $y$, arbitrarily, and make $y$ a pendant in~$T$. It remains to show that the distance between two adjacent vertices $v,w$ in $G$ is at most~$3$ in~$T$. i) $v,w \in X$: since we have a star in $T$ with respect to $X$, then $d(v,w)= 2$. ii) $v \in X$, $w \in Y$: the worst case occurs when 
$d_G(w)\geq 2$ and $v$ is a leaf of the star in $T$. In this case, $d(v,w) = 3$ by the path $vxx'w$, where $x'w$ belongs to $T$.
\end{proof}
 
In Proposition~\ref{prop:2Split} we characterize the stretch indexes for split graphs.
 
\begin{proposition}\label{prop:2Split}
Let $G=(X,Y)$ be a split graph which is not a tree. Thus, $\sigma_T(G)\!=\!2$ if and only if either: i) $d_G(y) = 1, \forall \ y \in Y$, or ii) $\exists \ x \in \bigcap_{y \in Y} N_G(y),\ x~\in~X$ such that $d_G(y) \geq 2$.

\end{proposition}
\begin{proof}
If $G$ satisfies i) or ii), then $G$ contains a tree $2$-spanner which can be constructed following Lemma~\ref{lm:upperSplit}, and, particularly in case ii), consider any vertex $x$ satisfying conditions required in ii) to be center of the star. Conversely, by contradiction, since $\sigma_T(G) = 2$, for each pair of vertices in $X$ there is in $T$ either an edge or a $P_3$ centered in a vertex $v$ of $G$. If $v \in X$, then the minimum stretch spanning subtree with respect to $X$ is a star. Otherwise, $v \in Y$ and each vertex of the clique would be a leaf of the star centered in $v$. Once there are two vertices in $Y$ with degree at least~$2$ without an adjacent vertex in common, in the first case, for any center of the star we have chosen regarding the clique's vertices, there is a vertex of the stable set such that all its neighbors are leaves of the star, which implies~$\sigma_T(G) \geq 3$. In the second case, $\sigma_T(G) \geq 3$ anyway, because, by hypothesis, there exist at least two more vertices in $Y$ with degree at least $2$, and they will be adjacent only to the leaves of the star centered in $v$.%
 \end{proof}

\begin{theorem}\label{thm:p4tidy-complex}
$2$-admissibility can be decided in $O(n)$ time for split graphs with $n$ vertices.
\end{theorem}
\begin{proof}
Given a split graph $G$, first we apply a preprocessing procedure by removing all pendant vertices of $G$, since these edges must belong to any tree $t$-spanner of $G$. 
The removal of such edges can be done in $O(n)$ time. 
After this step, we check the existence of a universal vertex, which can be done also in $O(n)$ by the degree sequence of $G$. 
\end{proof}




\subsection{$(0,2)$-Graphs}

$(0,2)$-Graphs are $3$-admissible, \emph{cf.}~\cite{TCS20}, and, in this work we characterize $2$-admissible $(0,2)$-graphs. 
Given a $(0,2)$-partition $X= (K^1,K^2)$ of a graph $G$, let $H_X = G[V_{H_X}]$ be $G$'s \emph{transversal subgraph with respect to $X$}, where $V_{H_X}$ is the set of vertices incident to each transversal edge of~$G$, i.e., edges with one extreme in $K^1$ and the other in $K^2$. 

\begin{lemma}\label{lm:02danca}
Let $G$ be a graph without a universal vertex and $X= (K^1,K^2)$ be a $(0,2)$-partition of $G$. If there exists a distinct $(0,2)$-partition for $G$, say $X'=(K^{1'}, K^{2'})$, then $H_X \simeq H_{X'} \simeq G$.

\end{lemma}

\begin{proof}
Let $G=(K^1 \cup K^2, E)$ be a $(0,2)$-graph and $H_X$ be its transversal subgraph with respect to $X$. Suppose there is a distinct $(0,2)$-partition for $G$, say $X'= (K^{1'},K^{2'})$. Since there is no universal vertex in $G$, there is $V' \subseteq K^1$ and $V''\subseteq K^2$ such that $(K^1 \setminus V') \cup V'' = K^{1'}$ and $(K^2 \setminus V'') \cup V' = K^{2'}$. Thus $N(V') \supseteq K^2 \setminus V''$ and $N(V'') \supseteq K^1 \setminus V'$. Therefore, there is a transversal edge incident to each vertex of $G$, and so $H_X \simeq G$. Moreover, for each $v\in V'$, $N_{K^{1'}}(v) \supseteq K^1 \setminus V'$ and for each $w \in V''$, $N_{K^{2'}}(w) \supseteq K^2\setminus V''$. Hence, there are transversal edges incident to all vertices of $G$ considering the partition $X'$, and consequently $H_{X'} \simeq G$.      
\end{proof}

By Lemma~\ref{lm:02danca}, we can assure that Lemma~\ref{lm:2adm02} provides the correct answer by the analysis of any $(0,2)$-partition. 

\begin{lemma}\label{lm:2adm02}
Let $G=(K^1\cup K^2, E)$ be a $(0,2)$-graph with the transversal subgraph of $G$ with respect to the $(0,2)$-partition $(K^1, K^2)$ given by $H$. $G$ is $2$-admissible if and only if either $G$ has a universal vertex, $G$ has a cut-vertex or $H$ is a strict $2$-connected graph that has not an induced $C_4$. 
\end{lemma}
\begin{proof}
Clearly, if $G$ has a universal vertex, then $\sigma_T(G) = 2$. If $G$ has a cut-vertex $v$, there is a bi-star with one center vertex in $v$ and the other in a neighbor of $v$ from the other clique, and $\sigma_T(G)=2$. If $H$ is a strict $2$-connected graph, $|N_{K^1}(K^2)| = 2$, w.l.g., and, once $H$ has not an induced $C_4$, there is a $C_4$ in $H$, say $abcda$ with chord $ac$, where $a,b \in K^1$. Since $H$ is strict $2$-connected, only $a$ can have other neighbors in $K^2$. We construct a tree $2$-spanner for $G$ which is a bi-star as follows: $a$ is one center of the bi-star whose leaves are $V(K^1) \setminus \{a\} \cup \{c\}$, and $c$ is the other center of the bi-star whose leaves are $V(K^2) \setminus \{c\}$. It is easy to see that $\sigma_T(G) = 2$.
For the converse, 
suppose $\sigma_T(G) = 2$, and, by contradiction, that there is not a universal nor a cut-vertex in $G$ and that $H$ is either a $3$-connected graph or is strict $2$-connected with an induced $C_4$. First, suppose $H$ is triconnected. Note that $G$ is not $3$-connected, because there is not a universal vertex in $G$. 
Since $\sigma_T(G) = 2$, there is a spanning tree $T$ of $G$ such that $T \cap H$ is a star. Whatever the center is, there is a vertex in $K^1$ or $K^2$ which is not adjacent to the center, once there is not a universal vertex in $G$, and the distance between this vertex and at least one of its neighbors of $G$ in $T$ is at least $3$. 
Now, suppose $H$ is strict $2$-connected with an induced $C_4$, say $abcda$. In this case, $G$ is clearly strict $2$-connected, and, w.l.g., $N_{K^1}(K^2)=\{a,b\}$. If $K^1$ and $K^2$ are both triconnected, then there is a spanning tree $T$ whose intersection with $K^1$ and $K^2$ are stars, thus $T$ is a bi-star and the distance in $T$ between two adjacent vertices of the $C_4$ is $3$. If both are strict $2$-connected, they are $K_3$'s, and we have $2$ possibilities: either both are treated as stars in the tree $2$-spanner, or at least one of them is a path. In both cases, the stretch index is at least $3$. If only one between $K^1$ and $K^2$ is triconnected, say $K^1$, then again there is a spanning tree $T$ whose intersection with $K^1$ is a star, and similarly to the previous arguments, the distance between two adjacent vertices of the $C_4$ is at least $3$ in~$T$. Thus, both cases lead us to contradictions.  
 \end{proof}

\begin{theorem}\label{thm:p4tidy-complex}
$2$-admissibility can be decided in $O(n+m)$ time for $(0,2)$-graphs with $n$ vertices and $m$ edges.
\end{theorem}
\begin{proof}
Assuming that $G$ is a $(0,2)$-graph, we check the three situations of Lemma~\ref{lm:2adm02}. 
Firstly, we check the existence of a universal vertex, which can be done in $O(n)$ time by the degree sequence of $G$. 
Secondly, in order to find a cut-vertex in $G$, we use a depth-first search algorithm, which can be done in $O(n+m)$ time. 
For the third condition, we use a depth-first search algorithm to check if $H$ is $2$-connected, where $H$ is the transversal subgraph of $G$ with respect to a $(0,2)$-partition, and, in order to detect the existence of an induced $C_4$ in $H$, we delete all non transversal edges of $G$. 
After this step, if it remains a $2K_2$ subgraph, it means that the $G$ contains an induced $C_4$. This strategy can be done in $O(n+m)$ time. 
\end{proof}

\subsection{$(0,\ell)$-graphs}

\cite{LAGOS19} settled that $t$-admissibility is \NP-complete for $(0,\ell)$-graphs with $n$ vertices even if $\ell$ is a linear function on $n$. In spite of that, bounds for the stretch index on $(0,\ell)$-graphs can be obtained, as stated in Theorem~\ref{thm:LowerUpper}.

The \emph{inflation of a graph $G$} is the graph $H$ obtained by replacing each vertex $v$ of $G$ by a clique of size $d(v)$, and each edge $uv$ of $G$ by an edge between two vertices of the corresponding cliques of $H$ in such a way that the edges of $H$ which come from the edges of $G$ form a matching of $H$~\cite{favaron1998irredundance}.

The class of graphs obtained by inflation is equivalent to the line graphs of subdivided graphs~\cite{roumuloTese}. 
Note that subdivided graphs is a subclass of bipartite graph. 
\cite{ctw20} proposed the \emph{edge $t$-admissibility} problem, which consists of solving $t$-admissibility for line graphs. 
The authors proved that \emph{edge $t$-admissibility} is \NP-complete even for bipartite graphs, i.e. $t$-admissibility is \NP-complete for line graphs of bipartite graphs.

A \emph{generalized inflation of a graph $G$} is a graph $H$ obtained by replacing each vertex $v$ of $G$ by a clique of size at least $d(v)$, and the remaining construction follows similarly to the inflation of a graph~\cite{alessandraDiss}. 
Note that any generalized inflation graph $H$ obtained by a graph $G$ is a $(0,\ell)$-graph, for values of $\ell = n(G)$.

Fig.~\ref{fig:inflation} depicts generalized inflation graphs with their corresponding tree spanners.

\begin{figure}[!h]
    \centering
    \includegraphics[width=9cm]{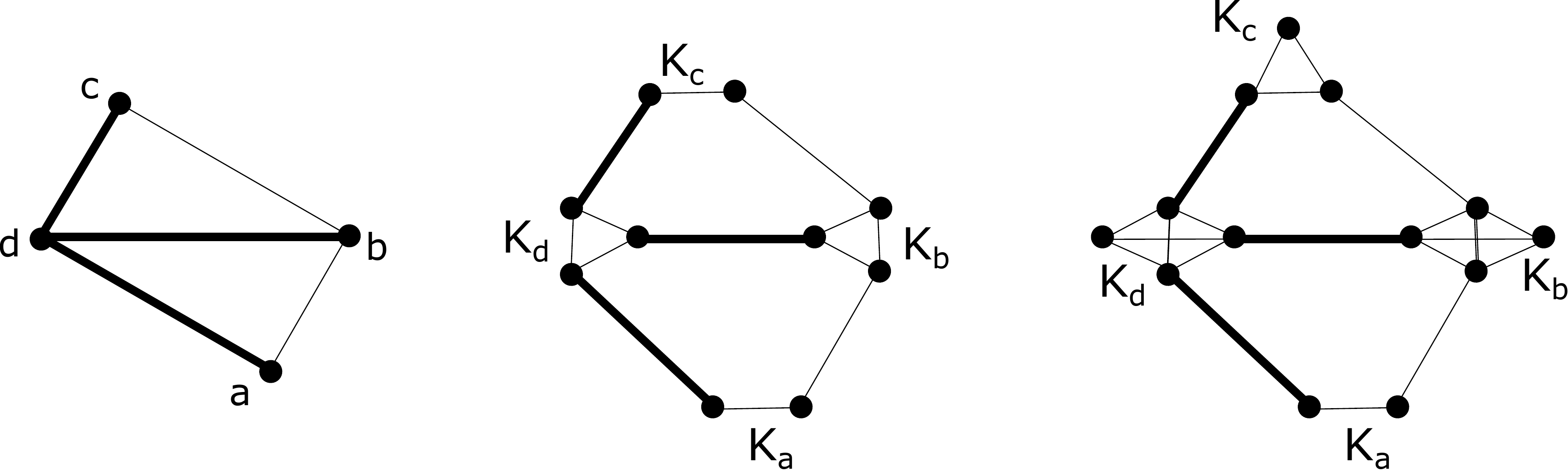}
    \caption{On the left: a connected graph $G$ whose bold edges correspond to a tree $2$-spanner $T$. On the middle and on the right: the inflation graph $H$ and a generalized inflation graph $Q$ of $G$, resp. Considering the last two graphs, bold edges of $H$ and $Q$ belong to their tree spanners from $T$. The remaining edges for the tree spanners of $H$ and of $Q$ must be stars centered in vertices adjacent to distinct cliques of graphs $H$ and $Q$.\label{fig:inflation}}
\end{figure}

\begin{lemma}\label{lm:inflado}
Any generalized inflation graph $H$ of a connected graph $G$ has stretch index equal to $\sigma(H) = 2\sigma(G) + 1$.
\end{lemma}
\begin{proof}
Consider a cycle of $G$ obtained by a sequence of vertices that corresponds to a longest $\sigma(G)$-path $P_{\sigma(G)+1} = u_1, u_2, \cdots, u_{\sigma(G) +1}$ of a tree $\sigma(G)$-spanner of $G$. 

Since each added clique of $H$ corresponding to a vertex $v$ in cycle of $G$ has size at least $2$ in $H$ because $d_G(v)\geq 2$,    
then $H$ has a cycle of size $2(\sigma(G)+1)$ obtained by a sequence of vertices $u^1_1, u^2_1, u^1_2, u^2_2, \cdots, u^1_{\sigma(G)+1}, u^2_{\sigma(G)+1}, u^1_1$, where  $u_i^1$ and $u_i^2$, for $i = 1, \cdots, \sigma(G)+1$, are two vertices of the clique corresponding to $u_i \in P_{\sigma(G)+1}$. 
Since inside an added clique of a longest path in any tree $\sigma(H)$-spanner of $H$ we must consider as few edges as possible, 
then we use a unique edge for each added clique in a longest path for any tree $t$-spanner of $H$. 
It implies that in a $(2\sigma(G)+1)$-path $P_{2(\sigma(G)+1)}$ in $T$. 
Hence, $\sigma(H) \geq 2\sigma(G)+1$. 

For the upper bound, 
note that each edge of a tree $\sigma(G)$-spanner of $G$ has a corresponding edge in $H$.
Hence, we construct the spanning tree $T$ of $H$ as follows: 
add to $T$ the edges of $H$ that correspond to the edges of a tree $\sigma(G)$-spanner of $G$; 
now, for each added clique $K$ of $H$ add to $T$ a spanning star of $K$ centered in a vertex $w$ such that $d_H(w)\geq |K|$.
Therefore, note that there is a path in $T$ given by $P_{2(\sigma(G)+1)} = u^1_1, u^2_1, u^1_2, u^2_2, \cdots, u^1_{2(\sigma(G)+1)}, u^2_{2(\sigma(G)+1)}$, implying $\sigma(H) \leq 2\sigma(G)+1$.
\end{proof}

\begin{corollary}\label{cor:cycle}
If $G$ is a cycle graph $C_{\ell}$, then any generalized inflation graph of $G$ has stretch index equal to $2\ell - 1$.
\end{corollary}
\begin{proof}
Let $H$ be any generalized inflation graph of $G = C_{\ell}$. 
Note that $\sigma(C_{\ell}) = \ell -1$. 
Hence, by Lemma~\ref{lm:inflado}, $\sigma(H) = 2(\ell -1) +1 = 2\ell -1$.
\end{proof}

\begin{corollary}\label{cor:npc1}
$t$-admissibility is \NP-complete for generalized inflation graphs.
\end{corollary}
\begin{proof}
By Lemma~\ref{lm:inflado}, $G$ is $t$-admissible if and only if its inflation graph $H$ is $(2t+1)$-admissible. Since $t$-admissibility is \NP-complete for an arbitrary connected graph $G$, then it is \NP-complete for its inflation graph $H$, which is equivalent to line graph of the subdivision of $G$, as proved by~\cite{alessandraDiss}.
\end{proof}

\begin{corollary}
$t$-admissibility is \NP-complete for line graphs of subdivided graphs.
\end{corollary}

Note that generalized inflation graphs are a proper subset of $(0,\ell)$-graphs. 
Besides of that, for any $(0,\ell)$-graph $H$ we can transform it into a graph $G$, that we call $G$ a \emph{subjacent graph of $H$} in the following way: by transforming each clique $K$ of the $(0,\ell)$ partition of $H$ into a unique vertex $u_K$ of $G$, and for each set of edge between cliques $K$ and $K'$ of $H$ we consider an edge $u_Ku_{K'}$ of $G$.

\begin{theorem}\label{thm:LowerUpper}
Let $G$ be a $(0,\ell)$-graph distinct of a tree. Hence, $2 \leq \sigma(G) \leq 2\ell-1$.
\end{theorem}
\begin{proof}
For the lower bound, note that any complete graph $K_{\ell'}$ is a $(0,\ell)$-graph, $\ell'\geq \ell$, and $\sigma(K_{\ell'})=2$. Hence, for any $(0,\ell)$-graph $G$, $\sigma(G)\geq 2$. 
For the upper bound, we first analyze the structure of any $(0,\ell)$-graph.
We can partition the set of $(0,\ell)$-graphs into two sets: i) the subset $A$ of the $(0,\ell)$-graphs which are generalized inflation graphs; ii) the subset $B$ of the $(0,\ell)$-graphs which are not generalized inflation graphs.
Case i) The generalized inflation graphs with maximum stretch index are those graphs $H$ obtained from graphs $G$ with maximum stretch index. Those graphs $G$ are the cycle graphs, and from Corollary~\ref{cor:cycle}, the stretch index for the graphs that belong to $A$ is upper bounded by $2\ell - 1$.
Case ii) The graphs that belong to $B$ are not inflation graphs. In this case, for any graph $H$ in $B$ we can consider its subjacent graph $G$.  Hence, we have that the stretch index of $H$ is equal to $2\sigma(G)+1$, for some subjacent graph $G$, by the same argument given in Lemma~\ref{lm:inflado}. 
Therefore, any subjacent graph $G$ isomorphic to a cycle graph correspond to a graph $H$ for which any edge of $H$ is between two vertices inside a clique of a $(0,\ell)$ partition, or it is between between two consecutive cliques associated to two consecutive vertices of the cycle graph $G$. 
The worst case occurs when the set of edges between two consecutive cliques form a matching, otherwise there would be a vertex $y_K$ in a clique $K$ adjacent to a vertex in consecutive cliques $K'$ and $K''$, hence it would be possible to create a spanning tree of $H$ only using $y_K$ on a $\sigma(H)$-path.
Hence, $H$ has $\sigma(H) = 2\sigma(G)+1 = 2\ell -1$ because $G$ is a cycle graph.
Suppose that a subjacent graph $G$ of $H$ is non isomorphic to a cycle graph, then there are edges between non consecutive vertices of a cycle, which correspond to a stretch index of $G$ less than $n(G)-1$, hence, less then the stretch index of a cycle graph, and by Lemma~\ref{lm:inflado}, $\sigma(H) < 2(n(G)-1) +1 = 2n(G)-1$.
\end{proof}

Note that as consequence of previous results, we have that if a graph is a $(0,\ell)$, then its stretch index is at most $2\ell -1$ and the equality holds if a graph is a generalized inflation graph of a cycle graph. 
One may ask if such a condition is also necessary, but the answer is no. 
An example of a $(0,\ell)$-graph that is not a generalized inflation graph of a cycle graph and has stretch index equal to $2\ell -1$ is a cycle power graph $C_6^2$, which can be constructed from a cycle $C_6$, adding edges from vertices at a distance equal to $2$ in $C_6$. 
Note that such a graph is a $(0,2)$-graph and $\sigma(C_6^2) = 3$, as proved by~\cite{TCS20} the stretch index of any cycle power graph.

As a direct consequence of the arguments given in Case ii) of Theorem~\ref{thm:LowerUpper}, we have Corollary~\ref{cor:charactUpper}.

\begin{corollary}\label{cor:charactUpper}
A $(0,\ell)$-graph $H$ has $\sigma(H) = 2\ell - 1$ if and only if the two following restrictions hold: i) the subjacent graph of $H$ is a cycle graph; and ii) there is not a vertex $w$ in $H$ belonging to a clique $K$ of a $(0,\ell)$ partition such that $w$ is adjacent to vertices belonging to two other cliques $K'$ and $K''$ of the $(0,\ell)$ partition of $H$.
\end{corollary}

\section{Conclusions}\label{sec:conc}

The $3$-admissibility has been studied for many years, and several classes are known to be $3$-admissible. Tracking these classes and structurally characterizing them sounds good when studying the tractability of the $3$-admissibility, which is the greatest problem we aim to solve. 
In this work we determine strategies for determining the stretch indexes' for split graphs and cographs, known to be $3$-admissible~(c.f.~\cite{brandstadt2007tree,TCS20}), and we deal with their superclasses: $(k,\ell)$-graphs and graphs with few $P_4$'s. 
As a byproduct, we prove that $t$-admissibility is \NP-complete for line graphs of subdivided graphs, corresponding to an advance regarding the edge $t$-admissible problem, proposed by~\cite{ctw20}.
Regarding $(0,\ell)$-graphs, we present lower and upper bounds for the stretch index, in such a way that we are able to determine graph classes that have stretch indexes reaching either the lower or the upper bound and characterize graphs that have stretch index equal to the proposed upper bound. 
For graphs with few $P_4$'s, our characterizations yield to simpler and faster algorithms in order to determine their stretch indexes. 
An interesting question is: Is it possible to develop faster algorithms for some other graph classes already classified according to their admissibility? 

\section*{Acknowledgments}
This study was financed by CAPES - Finance Code 001, and by FAPERJ.




\end{document}